\documentclass[preprint,12pt,english]{elsarticle}
\usepackage[T1]{fontenc}
\usepackage[latin9]{inputenc}
\usepackage{amsmath}
\usepackage{graphicx}
\usepackage{amssymb}
\usepackage{esint}
\usepackage{amsthm}
\usepackage{dcolumn}
\usepackage{bm}
\usepackage{babel}

\begin{document}
\newtheorem{conjecture}{Conjecture}\newtheorem{corollary}{Corollary}\newtheorem{theorem}{Theorem}
\newtheorem{lemma}{Lemma}\newtheorem{observation}{Observation}\newtheorem{definition}{Definition}
\newtheorem{remark}{Remark}\global\long\global\long\def\ket#1{|#1 \rangle}
 \global\long\global\long\def\bra#1{\langle#1|}
 \global\long\global\long\def\proj#1{\ket{#1}\bra{#1}}

\begin{frontmatter}

\title{The classicality and quantumness of a quantum ensemble}

\author[hunan,ustc]{Xuanmin Zhu}

\author[ustc]{Shengshi Pang}

\author[ustc]{Shengjun Wu\corref{cor}}

\ead{shengjun@ustc.edu.cn}

\author[hunan]{Quanhui Liu}

\ead{qhliu@hnu.cn}

\address[hunan]{School for Theoretical Physics and Department of Applied Physics
Hunan University, Changsha 410082, China}

\address[ustc]{Hefei National Laboratory for Physical Sciences at Microscale and
Department of Modern Physics, University of Science and Technology
of China, Hefei, Anhui 230026, China}

\cortext[cor]{Corresponding author.}

\begin{abstract}
In this paper, we investigate the classicality and quantumness of
a quantum ensemble. We define a quantity called ensemble classicality
based on classical cloning strategy (ECCC) to characterize how classical
a quantum ensemble is. An ensemble of commuting states  has a unit
ECCC, while a general ensemble can have a ECCC less than 1. We also
study how quantum an ensemble is by defining a related quantity called
quantumness. We find that the classicality of an ensemble is closely
related to how perfectly the ensemble can be cloned, and that the
quantumness of the ensemble used in a quantum key distribution (QKD)
protocol is exactly the attainable lower bound of the error rate in
the sifted key.\end{abstract}

\end{frontmatter}

\section{Introduction}
Quantum theory has revealed many counterintuitive features of quantum
systems in comparison with those of classical systems. The state of a classical system
can be copied, deleted or distinguished with a unit probability,
while an unknown quantum state can never be perfectly copied or deleted ~\cite{clone,clone0,delete},
and non-orthogonal quantum states
cannot be reliably distinguished~\cite{nielsen,pang}. The no-cloning
theorem assures the security of quantum key distribution protocols~\cite{bb84}
and prohibits superluminal communication\cite{clone3}. Non-commuting
observables in quantum mechanics cannot be determined simultaneously,
and a quantum measurement usually disturbs the involved quantum systems,
in striking contrast to the fact that measurements can leave classical
systems unperturbed in principle.

In this paper, we study the classicality and quantumness of a quantum
ensemble $\mathcal{E}=\{q_{i},\rho_{i}\}$,  specified by
the set of states $\rho_{i}$ and the corresponding probabilities $p_{i}$.
Some quantum ensembles can be manipulated
like classical ones, whereas others can not. For example, an unknown
state from an ensemble $\mathcal{E}_{ort}$ consisting of orthogonal
pure states could be cloned perfectly and determined without being
disturbed; on the other hand, a state from an ensemble $\mathcal{E}_{non}$
consisting of non-orthogonal states cannot be cloned perfectly and determined
exactly~\cite{extraction}. By classicality, we mean how well
a quantum ensemble can be manipulated as a classical one. Perfect
clonability and distinguishability are essential characteristics of
classical sets of states. Intuitively, the ensemble $\mathcal{E}_{ort}$
is more classical than $\mathcal{E}_{non}$, so the following questions
naturally arise: what kind of ensembles could be handled like classical ones
and what kind could not? Is there a quantity to quantify how classical
an ensemble is? There have already been some researches on the quantumness
of quantum ensembles~\cite{qianren1,qianren2,qianren3,luo}. In
this paper, we study the classicality and quantumness of a quantum
ensemble from a different perspective. We start from considering how
precisely an unknown state from the ensemble can be cloned and how
stable it is under an appropriate measurement, i.e., how close the
state after the measurement is to the original one.

For an arbitrary unknown input state $\rho$,
a universal perfect cloning process does not exist,
and many approximate cloning strategies have been proposed. One interesting
strategy is given by the unitary transformation
 $\ket j\ket 0\to\ket j\ket j$, where $\{\left|j\right\rangle \}$
is a basis  of the Hilbert space of the input system and $\ket 0$
is a blank state of an ancillary system. This cloning strategy was
first introduced in \cite{clone}, and we call it a \emph{classical
cloning strategy} under basis $\{\left|j\right\rangle \}$
as it is the quantum counterpart of the cloning process in the classical
world.

Obviously, this classical cloning strategy is neither perfect nor
optimum for cloning an unknown quantum state. The copies produced
are generally different from the original state, so it is meaningful
to quantify the distance between a copy and the original state. The
way to measure the distance is investigated intensively and many proposals
have been put forward~\cite{nielsen,vedral1}. One distance measure
is the relative entropy~\cite{luo,vedral1}, which has been used
to quantify entanglement and correlations~\cite{vedral1,vedral,modi}. However,
the relative entropy is not a genuine metric as it is not symmetric.
Two other widely used distance measures, the trace distance and the
fidelity~\cite{nielsen}, are well defined because both of them
are symmetric and satisfy the requirements of good distance measures.
In this paper, we use fidelity as the distance measure. The fidelity
of $\rho$ and $\sigma$ is defined as~\cite{jozsa} \begin{equation}
F(\rho,\sigma)=(\mathrm{tr}{\sqrt{{\rho}^{1/2}\sigma{\rho}^{1/2}}})^{2}.\label{eq:fidelity}\end{equation}
(The square root of the above quantity is also frequently used as
the fidelity ~\cite{nielsen}, but we adopt Eq. \eqref{eq:fidelity}
as the fidelity definition throughout this paper.) It is obvious that
$0\leq F(\rho,\sigma)\leq1$ and $F(\rho,\sigma)=1$ if and only if
$\rho=\sigma$.

\section{The ensemble classicality based on classical cloning strategy}

For an ensemble  $\mathcal{E}=\{q_{i},\rho_{i}\}$ consisting of the set of
states $\{ \rho_i \}$ and the corresponding probabilities of occurrence $\{p_{i}\}$,
we investigate its classicality by studying how well an unknown state
from the ensemble can be cloned by the classical cloning strategy
under the basis $\{\ket j\}$. First, we define the \emph{average
cloning fidelity} for the ensemble $\mathcal{E}$ as \begin{equation}
F_{ave}(\mathcal{E},\{\ket j\})=\sum_{i}q_{i}F(\rho_{i},\rho_{i}^{\prime}),\label{eq:average}\end{equation}
where $\rho_{i}^{\prime}$ is the state of an output copy via the classical cloning
strategy under basis $\{\left|j\right\rangle \}$ if the input state is $\rho_i$, i.e.,\begin{equation}
\rho_{i}^{\prime}=\mathrm{Tr}_{2}\sum_{i,j}\langle i|\rho_{i}|j\rangle|i\rangle|i\rangle\langle j|\langle j|=\sum_{j}\bra j\rho_{i}\ket j\proj j.\end{equation}
For an ensemble of orthogonal pure states, an average cloning
fidelity  1 could be reached only if the states in the ensemble are actually
the cloning basis states.  For a general quantum ensemble
$\mathcal{E}=\{q_{i},\rho_{i}\}$, it can be seen that $F_{ave}\leq1$
as $F(\rho_{i},\rho_{i}')\leq1$ and $\sum_{i}q_{i}=1$. The average
cloning fidelity $F_{ave}$ represents the performance of a classical
copying strategy on a quantum ensemble; meanwhile, $F_{ave}$ can
also represent stability of the states in an ensemble under a projective
measurement, since $\rho_{i}^{\prime}=\sum_{j}\bra j\rho_{i}\ket j\proj j$
is also the density matrix after a von Neumann measurement on $\rho_{i}$
along the basis $\{\ket j\}$. In this sense, the average cloning fidelity
characterizes how classical the ensemble is. Therefore, we define
a quantity $J$, the ensemble classicality based on classical
cloning strategy (ECCC), to quantify how classical the quantum ensemble
$\mathcal{E}=\{q_{i},\rho_{i}\}$ is,
\begin{equation}
J(\mathcal{E})=\max_{\{\ket j\}}\{F_{ave}(\mathcal{E},\ket j)\}=\max_{\{\ket j\}}\{\sum_{i}q_{i}F(\rho_{i},\rho_{i}')\},\label{eq:jingdian}
\end{equation}
where $\{\ket j\}$ is an orthonormal basis of the subspace spanned
by the states in the ensemble. For an infinite quantum ensemble $\mathcal{E}=\{f(\alpha),\rho(\alpha)\}$,
the ECCC is similarly defined as \begin{equation}
J(\mathcal{E})=\max_{\{\ket j\}}  \int f(\alpha)F(\rho(\alpha),\rho(\alpha)'){\mathrm{d}}\alpha,\label{eq:lianxu}\end{equation}
where $\rho(\alpha)'=\sum_{j}\bra j\rho(\alpha)\ket j\proj j$ is the state of an output copy
for an input $\rho (\alpha)$ and
$f(\alpha)$ is the probability distribution function satisfying $\int f(\alpha){\mathrm{d}}\alpha=1$.
It can be seen that the $J(\mathcal{E})$ defined above is an intrinsic
property of the ensemble, independent of the cloning basis. It is
evident that $\mathcal{E}$ can be manipulated like a classical
ensemble only if $J(\mathcal{E})=1$.

A single state $\rho$ can be considered as an ensemble consisting
of just one state with unit probability.
The ECCC of a single-state ensemble $\rho$ is equal to one, since
the cloning basis states could be chosen as the eigenstates of $\rho$,
then $\rho=\rho'$, and thus $J=F(\rho,\rho')=1$.

In the following, the properties of ECCC will be studied. The range of $J$ will be given
in the first two theorems.

\begin{theorem} \label{thm:range}
The ECCC of a general ensemble $\mathcal{E}$ of states in a $d$-dimensional Hilbert space has the following upper and lower bounds:
(i) $J(\mathcal{E})\leq 1$, with $J(\mathcal{E})=1$ if and only if all quantum
states in the ensemble commute with each other; and (ii)
$J(\mathcal{E})> \frac{1}{d}$ for any ensemble $\mathcal{E}$, and
$J(\mathcal{E})\geq \frac{1}{d}+ q_{m} \frac{d-1}{d} \geq \frac{N+d-1}{Nd}$ for any finite ensemble $\mathcal{E}=\{q_i, \rho_i | i=1,2,\cdots,N\}$ of $N$ states,
where $q_{m}=max\{q_1,\cdots,q_N \}$.
\end{theorem}

Proof of theorem 1 is given in appendix A.
The inequality $\frac{1}{d} < J(\mathcal{E}) \leq 1$ is also valid for an ensemble of infinite number of states.  The lower bound $\frac{1}{d}$ is generally not achievable for finite or infinite ensembles.
Before presenting attainable lower bounds for specific cases, we give the following lemma (its proof is given in appendix B), which will be used in proving theorem 2.

\begin{lemma}
For any state $\rho$ of a qubit system, the classical cloning strategy
is performed with respect to a basis $\{\ket{e_{i}}\}$ and the state
of either output copy is denoted by $\rho'$, we have the following inequality \begin{equation}
F(\rho,\rho')\geq\sum_{i=0}^{1}q_{i}F(\ket{\psi_{i}},\rho_{i}')=\sum_{i,j=0}^{1}q_{i}|\langle e_{j}|\psi_{i}\rangle|^{4},\end{equation}
where $\rho_{i}'=\sum_{j=0}^{1}|\langle e_{j}|\psi_{i}\rangle|^{2}\proj{e_{j}}$ is
the state of an output copy for an input $\ket{\psi_{i}}$, and $\{q_{i},\ket{\psi_{i}}\}$ are the eigenvalues and eigenvectors of $\rho$.
Here, the right hand side of the inequality is actually the average cloning fidelity of
the eigen-ensemble $\mathcal{E} = \{q_i, \ket{\psi_{i}} \}$ of $\rho$.
\end{lemma}

In the following theorem, we present a tighter and achievable lower bound  ($\frac{2}{d+1}$) of the ECCC for two special cases.

\begin{theorem} \label{thm:range1}
(i) For an ensemble $\mathcal{E}$ of pure states in a $d$-dimensional Hilbert space, $\frac{2}{d+1}\leq J(\mathcal{E})\leq1$.

(ii) For an ensemble $\mathcal{E}$ consisting of general (pure or mixed) states in a two-dimensional Hilbert space of a qubit system, $\frac{2}{3}\leq J(\mathcal{E})\leq1$. \end{theorem}

The proof is given in appendix C.
The lower bounds are actually achieved by an infinite ensemble consisting equiprobably of all pure states in the respective Hilbert space (see appendix C).

The ECCC $J$ of an ensemble $\{q_{i},\rho_{i}\}$ quantifies the maximum average performance of cloning the states from the ensemble by a classical strategy, thus provides a measure of how classical the ensemble is. From another
perspective, the quantity $J$ of an ensemble $\{q_{i},\rho_{i}\}$
also tells us to what extent the states in the ensemble commute. The
ECCC of an ensemble of mutually commuting states is equal to $1$,
this is also in accordance with the fact that commuting states could be
broadcasted~\cite{broadcast}.

\begin{theorem} \label{thm:zhiji} For the ensembles $\mathcal{E}_{A}=\{q_{i},\rho_{iA}\}$,
$\mathcal{E}_{B}=\{q_{j},\rho_{jB}\}$, and $\mathcal{E}_{AB}=\{{q_{i}}{q_{j}},\rho_{iA}\otimes\rho_{jB}\}$,
there is an inequality \begin{equation}
J(\mathcal{E}_{AB})\geq J(\mathcal{E}_{A})J(\mathcal{E}_{B});\label{eq:1}\end{equation}
the inequality \eqref{eq:1} is also valid
for the infinite ensembles $\mathcal{E}_{A}=\{f(\alpha),\rho_{A}(\alpha)\}$,
$\mathcal{E}_{B}=\{f(\beta),\rho_{B}(\beta)\}$, and $\mathcal{E}_{AB}=\{f(\alpha)f(\beta),\rho_{A}(\alpha)\otimes\rho_{B}(\beta)\}$.
\end{theorem} \begin{proof}
Assume that $\{\ket k _A\}$ and $\{\ket m _B\}$ are the bases of the systems
$A$ and $B$ which maximize $\sum_{i}q_{i}F(\rho_{iA},\rho_{iA}')$
and $\sum_{j}q_{j}F(\rho_{jB},\rho_{jB}')$ respectively, then $J(\mathcal{E}_{A})=\sum_{i}q_{i}F(\rho_{iA},\rho_{iA}')$,
$J(\mathcal{E}_{B})=\sum_{j}q_{j}F(\rho_{jB},\rho_{jB}')$, where
$\rho_{iA}'=\sum_{k}\bra k\rho_{iA}\ket k\proj k$ and $\rho_{jB}'=\sum_{m}\bra m\rho_{jB}\ket m\proj m$.
The basis $\{\ket k\otimes\ket m\}$ may not be optimal for $\mathcal{E}_{AB}$,
so from the definition of $J$ we can get \begin{equation}
\begin{split}J(\mathcal{E}_{AB}) & =\max_{\{\ket l^{AB}\}}\{F_{ave}(\mathcal{E}_{AB},\{\ket l^{AB}\})\}\\
 & \geq F_{ave}(\mathcal{E}_{AB},\{\ket k\otimes\ket m \})\\
 & =\sum_{ij}q_{i}q_{j}F(\rho_{iA}\otimes\rho_{jB},\rho_{iA}'\otimes\rho_{jB}')\\
 & =\sum_{ij}q_{i}q_{j}F(\rho_{iA},\rho_{iA}')F(\rho_{jB},\rho_{jB}')\\
 & =J(\mathcal{E}_{A})J(\mathcal{E}_{B}).\end{split}
\end{equation}
The proof for infinite ensembles is similar. \end{proof}

In fact, we have not found any example for which $J(\mathcal{E}_{AB})$ is strictly greater than $J(\mathcal{E}_{A})J(\mathcal{E}_{B})$
so far, so it is an open question that whether $J(\mathcal{E}_{AB})=J(\mathcal{E}_{A})J(\mathcal{E}_{B})$
holds true for all ensembles $\mathcal{E}_{A}$, $\mathcal{E}_{B}$,
$\mathcal{E}_{AB}$ defined in Theorem 3.

It is intuitive to suggest that for an
arbitrary ensemble $\{p_{i},\rho_{i}\}$ and a standard state $\proj 0$,
there is an inequality $J(\{p_{i},\rho_{i}\otimes\proj 0\})\geq J(\{p_{i},\rho_{i}\otimes\rho_{i}\})$,
with equality if and only if all $\rho_{i}$ are commuting.
However, we don't know how to prove this conjecture.

We show that $J$ is invariant under unitary operations. For
a finite ensemble $\mathcal{E}=\{q_{i},\rho_{i}\}$, after a unitary
operation $U$, the ECCC of the new ensemble is given as \begin{equation}
\begin{split}J(U\mathcal{E}U^{\dagger}) & =\max_{\{\ket j\}}\{F_{ave}(U\mathcal{E}U^{\dagger},\ket j)\}\\
 & =\max_{\{U\ket j\}}\{F_{ave}(U\mathcal{E}U^{\dagger},U\ket j)\}\\
 & =\max_{\{U\ket j\}}\{\sum_{i}q_{i}F(U\rho_{i}U^{\dagger},\sum_{j}\bra j\rho_{i}\ket jU\proj jU^{\dagger})\}\\
 & =\max_{\{\ket j\}}\{\sum_{i}q_{i}F(\rho_{i},\rho_{i}')\}\\
 & =J(\mathcal{E}).\end{split}
\label{eq:2}\end{equation}
 It is obvious that the above equality is also valid for infinite
ensembles. Therefore, an ensemble $\mathcal{E}_{0}$ can be transformed
to another ensemble $\mathcal{E}_{1}$ by a unitary operation only
if they have the same ECCC, i.e., $J(\mathcal{E}_{0})=J(\mathcal{E}_{1})$.

As an example, we consider the set of states used in the BB84 protocol, and for any given $p$
($0 \leq p \leq 1$) we define
an ensemble $\mathcal{E}(p)$ as the set of states $\{\ket{0}, \ket{1}, (\ket{0}+\ket{1})/\sqrt{2},(\ket{0}-\ket{1})/\sqrt{2}\}$ with
different prior probabilities $\{p/2,p/2,(1-p)/2,(1-p)/2\}$.
The ensemble used in the BB84 protocol is essentially $\mathcal{E}(p=0.5)$.
A straightforward calculation yields the ECCC of the ensemble $\mathcal{E}(p)$ as
\begin{equation}
J(\mathcal{E}(p)) = \frac{3}{4} + \frac{1}{4} |2 p -1| .
\end{equation}
For the two ensembles, $J(\mathcal{E}(0.9))$ and $J(\mathcal{E}(0.5))$, specified by two different values of $p$, one easily has $J(\mathcal{E}(0.9))=0.95$, and $J(\mathcal{E}(0.5))=0.75$.  Although both ensembles include the same set of quantum states, the ensemble $\mathcal{E}(0.9)$ is much more classical than the ensemble $\mathcal{E}(0.5)$ according to our definition of classicality.
This is also intuitively correct, as in the limit case $p \rightarrow 0$ or $p \rightarrow 1$,
the ensemble $\mathcal{E}(p)$ becomes a purely classical ensemble.
The ECCC $J$ of an ensemble $\mathcal{E} = \{ p_i, \rho_i \}$ is essentially the maximum average cloning fidelity under a classical cloning strategy, it depends on the set of prior probabilities $\{p_i \}$.

Next, we consider two specific ensembles with infinite number of states in two-dimensional Hilbert
space. A general basis of the two dimensional Hilbert space can be conveniently written as: $\ket{e_{1}}=\cos(\theta_{1}/2)\ket 0+\sin(\theta_{1}/2)e^{i\varphi_{1}}\ket 1$ and $\ket{e_{2}}=\sin(\theta_{1}/2)\ket 0-\cos(\theta_{1}/2)e^{i\varphi_{1}}\ket 1$. The first infinite ensemble $\mathcal{E}_{bloch}$ we consider consists of pure states uniformly distributed on the Bloch sphere, i.e., $\mathcal{E}_{bloch}=\{1/{4\pi},\cos(\theta/2)\ket 0+\sin(\theta/2)e^{i\varphi}\ket 1\}$, where $\theta\in[0,\pi]$ and $\varphi\in[0,2\pi)$.
The average cloning fidelity of this ensemble is $F_{ave}=2/3$ which
is independent of the basis used in the classical cloning process,
so $J(\mathcal{E}_{bloch})=2/3$.
The other ensemble we consider, a symmetric double-circle ensemble, is defined for a fixed $\theta$ as  $\mathcal{E}(\theta)=\{1/{4\pi},\cos(\theta/2)\ket 0\pm\sin(\theta/2)e^{i\varphi}\ket 1\}$, where $\varphi\in[0,2\pi)$. The states in the ensemble $\mathcal{E}(\theta)$
lie on two symmetric latitudinal circles of the Bloch sphere with
polar angles $\pm\theta$. The average cloning fidelity of this ensemble
is $F_{ave}(\theta,\theta_{1},\varphi_{1})=1-{\sin^{2}\theta}/{2}+{\sin^{2}\theta_{1}}(3{\sin^{2}\theta}-2)/4$.
According to the definition of $J$, we have \begin{equation}
\begin{split}J(\theta)= & \max_{\{\theta_{1},\varphi_{1}\}}\{F_{ave}(\theta,\theta_{1},\varphi_{1})\}\\
= & \left\{ \begin{aligned} & 1-\frac{1}{2}{\sin^{2}\theta} & \textrm{if \ensuremath{0\leq\sin\theta\leq\sqrt{2/3}}}\\
 & \frac{1}{2}+\frac{1}{4}{\sin^{2}\theta} & \textrm{if \ensuremath{\sqrt{2/3}<\sin\theta\leq1}}\end{aligned}.
\right.\end{split}
\label{eq:3}\end{equation}
It can be seen that when $\theta=\arcsin(\sqrt{2/3})$ or $\pi-\arcsin(\sqrt{2/3})$,
$J(\theta)$ reaches the minimal value $2/3$, which is also the
ECCC of $\mathcal{E}_{bloch}$. When $\theta=\pi/2$, the states in
the ensemble are equiprobably distributed on the $x-y$ equator, and
the $J$ of this ensemble is $3/4$.

An unknown state cannot be perfectly cloned, but can be approximately
cloned. The approximate cloning theories have been established and
developed very well~\cite{clone3,clone2,clone1,clone4}. In Fig.
1, the ECCC $J(\theta)$ of $\mathcal{E}(\theta)$ is depicted as a function of $\theta$, together with the ECCC $J(\mathcal{E}_{bloch})$ of the ensemble $\mathcal{E}_{bloch}$,
the fidelity of the optimal mirror phase-covariant cloning (MPCC)~\cite{clone4},
and the fidelity of universal cloning (UC)~\cite{clone1}.
From Fig. 1, one can see that the MPCC fidelity $F(\theta)$
and the $J(\theta)$ reach their minimal values ($5/6$ and $2/3$ respectively) simultaneously when $\theta=\arcsin(\sqrt{2/3})$ or $\theta=\pi-\arcsin(\sqrt{2/3})$.
The minimal value of the MPCC fidelity is equal to the UC fidelity, and
the minimal value of $J(\theta)$ is equal to $J(\mathcal{E}_{bloch})$.
Roughly speaking, Fig. 1 shows that the more classical
an ensemble is, the more perfectly the states in the ensemble can be cloned.
The ECCC $J$ of the ensembles used in the BB84 \cite{bb84} protocol and
in the six-state protocol \cite{six1,six2} are  $3/4$ and $2/3$
respectively. It is interesting to note that the optimal cloning strategies
for the BB84 ensemble and the six-state ensemble are equivalent to
the optimal strategies for the phase-covariant cloning and the universal cloning
respectively~\cite{clone3}.
\begin{figure}[t]
 \centering \includegraphics[scale=0.5]{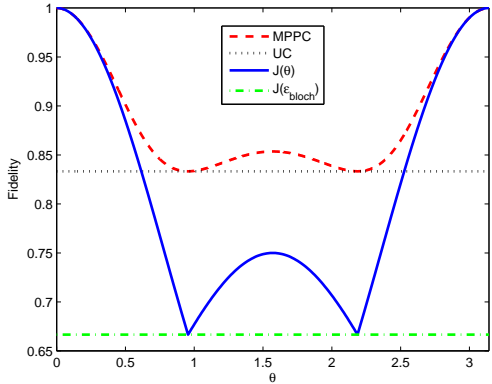} \caption{(Color online)The $\theta$ dependence of the ECCC and the cloning
fidelilties for $\mathcal{E}(\theta)$ and $\mathcal{E}_{bloch}$:
$J(\theta)$ (solid), $J(\mathcal{E}_{bloch})$ (dashdot), the MPPC
(dashed), and the UC (dotted). }

\label{fig:duibi}
\end{figure}

\section{The quantumness of an ensemble}

Next, we turn to study an opposite property of a quantum ensemble.
We define the quantumness $Q$ of an ensemble $\mathcal{E}=\{q_{i},\rho_{i}\}$
as \begin{equation}
Q(\mathcal{E})=1-J(\mathcal{E})=\min_{\{\ket j\}}\{\sum_{i}q_{i}(1-F(\rho_{i},\rho_{i}'))\}.\label{eq:quantumness}\end{equation}
The quantity $Q$ has similar properties to those of $J$. We have $0\leq Q<\frac{d-1}{d}$ for any
ensemble, and $0\leq Q\leq (1-q_m)(1-\frac{1}{d}) \leq \frac{(N-1)(d-1)}{Nd}$ for any finite ensemble $\mathcal{E}=\{q_i, \rho_i | i=1,2,\cdots,N\}$ of $N$ states,
where $q_{m}=max\{q_1,\cdots,q_N \}$.
It can also be seen that $Q=0$ for a single-state ensemble or an ensemble of
mutually commuting states, $Q\leq\frac{d-1}{d+1}$ for an ensemble of pure states, and $Q\leq\frac{1}{3}$ for an ensemble of states in a two-dimensional Hilbert space.

In~\cite{qianren2}, Fuchs et al. define a quantity $Q(S)$ to quantify
the quantumness of a set of pure states by the difficulty
of transmitting the states through a classical communication channel in the worst case.
Recently, Luo et al. gave a quantity $Q_{D}$
to quantify the quantumness of an ensemble through the disturbances
induced by von Neumann measurement~\cite{luo}. Instead of the relative
entropy which they used as the distance measure, we use the fidelity
to measure the distance between the two states.
Although both $Q_{D}$ and $Q$ are zero for the ensembles consisting of commuting states,
they are different in general.


The quantumness $Q$ of an ensemble tells us the extent to which the
ensemble is distinct from a purely classical ensemble, and we shall
see that the quantumness of an ensemble used for quantum key distribution
(QKD) is precisely the attainable lower bound of the error rate. In
the quantum key distribution theory, the error rate is the rate of
errors caused by eavesdroppers~\cite{qkd,ir}. Legitimate users
can use it to detect whether there exist eavesdroppers. Now we study
the relation between the quantumness of the ensemble used in a QKD
protocol and the error rate under the intercept-resend eavesdropping
strategies~\cite{ir}.
\begin{theorem}\label{thm: error rate}
The quantumness $Q$ of the ensemble used in a general QKD protocol
is the attainable lower bound of the error rate under the intercept-resend
eavesdropping strategy. \end{theorem}
\begin{proof} In a general
QKD protocol, Alice sends a pure state $\ket{\psi_{i}}$  to Bob with
a probability $q_{i}$, and the ensemble used is $\{q_{i},\ket{\psi_{i}}\}$.
When Bob's measurement basis is different from Alice's sending basis,
the state Bob receives is discarded, and when their bases are the
same, the received state is reserved. The measurement results of the
reserved states are usually called the \emph{sifted} keys. The error
rate is the average probability that Bob's measurement gives a wrong
result in the sifted key. With the intercept-resend strategy, the eavesdropper Eve
intercepts a state from Alice, say $\ket{\psi_{i}}$, then performs
a projective measurement along the basis $\{|j\rangle\}$ and gets
an output $\ket j$ with a probability $|\langle j|\psi_{i}\rangle|^{2}$,
and finally resends the output state to Bob. When Bob's measurement
basis is in accordance with Alice's sending basis, the probability
that Bob gets the original state $\ket{\psi_{i}}$ is $P=\sum_{j}|\langle j|\psi_{i}\rangle|^{4}=F(\proj{\psi_{i}},\rho_{i}')$,
where $\rho_{i}'=\sum_{j}|\langle j|\psi_{i}\rangle|^{2}\proj j$.
Thus the error rate for this strategy is $R=\sum_{i}q_{i}(1-F(\proj{\psi_{i}},\rho_{i}'))$.
The quantumness of the ensemble $\{q_{i},\ket{\psi_{i}}\}$ is $Q=\min_{\{\ket j\}}\{\sum_{i}q_{i}(1-F(\proj{\psi_{i}},\rho_{i}'))\}\leq R$.
Therefore, the quantumness $Q$ is the lower bound of the
error rate of a general QKD protocol, and the lower bound is achieved when the basis
along which Eve performs the measurement is chosen as the basis that is used to achieve the ECCC of the ensemble $\{q_{i},\ket{\psi_{i}}\}$.
\end{proof}

It is obvious that an ensemble whose quantumness $Q$ is zero or
very small is not suitable for QKD, since the eavesdropper
can get the information of the keys without being detected.
The quantumness $Q$ of an ensemble is closely related to the
security of QKD protocol against the intercept-resend eavesdropping
strategy.
The error rates for BB84 protocol and six-state protocol
are $1/4$ and $1/3$ respectively~\cite{qkd}.
By simple calculation, we know that the quantumness of the two
ensembles used in these two QKD protocols are $1/4$ and $1/3$ respectively, which
are equal to their error rates. The quantumness of the six-state ensemble
is $1/3$ which reaches the upper bound of the quantumness over all
ensembles of qubit states. For the intercept-resend eavesdropping
strategy, it can be seen that the six-state QKD protocol is most secure
among the QKD protocols which use states in two-dimensional Hilbert space.

\section{Conclusion}

In conclusion, we have proposed a quantity $J$, the
ensemble classicality based on classical cloning strategy (ECCC), to measure the
classicality of a given ensemble. The quantity $J$ can tell how classical
an ensemble is. When $J=1$ the ensemble behaves like a purely classical
ensemble; and when $J<1$ the ensemble cannot be considered as a classical
ensemble. We have revealed that the more classical an ensemble is,
the better an unknown state from the ensemble can be cloned. The quantity
of ECCC provides us with a tool to evaluate how well classical tasks
such as cloning, deleting, and distinguishing could be accomplished
for quantum ensembles. We also define the quantumness of an ensemble
and we surprisingly find that the quantumness of an ensemble used
in quantum key distribution is exactly the attainable lower bound
of the error rate. Our work could be useful for further investigation
of classical and quantum features of quantum ensembles and it could provide
a quantitative framework for various tasks in quantum communication.

\section*{Acknowledgments}

The authors acknowledge the support from the NNSF of China (Grant
No. 11075148), the CUSF, the CAS, and the National Fundamental Research
Program.

\section*{Appendix A: Proof of theorem 1}

\begin{proof}
(i) The upper bound can be easily shown since
$J(\mathcal{E})=\max_{\{\ket j\}}\{\sum_{i}q_{i}F(\rho_{i},\rho_{i}')\}\leq\sum_{i}q_{i}=1$
due to  $F(\rho_{i},\rho_{i}')\leq1$.
Now we prove that $J(\mathcal{E})=1$ if and only if all quantum states
in the ensemble are mutually commutative. Suppose $\{\ket{j^{*}}\}$
is the basis that maximizes $F_{ave}$ for the ensemble $\mathcal{E}$.
If $J(\mathcal{E})=1$, then for each $i$, $F(\rho_{i},\rho_{i}')=1$
and thus $\rho_{i}=\rho_{i}'=\sum_{j}\bra{j^{*}}\rho_{i}\ket{j^{*}}\proj{j^{*}}$.
So all the states are diagonal in the same basis $\{\ket{j^{*}}\}$,
and they commute with each other. On the other hand, if all the states
in the ensemble commute with each other, all of them can be diagonalized
simultaneously, i.e., there exists a basis in which all the states
are diagonal and we can use this basis in the classical cloning strategy,
then $\rho_{i}'=\rho_{i}$ and $F(\rho_{i},\rho_{i}')=1$ for each
$i$, so we get $J(\mathcal{E})=1$.

(ii) Now we try to prove the lower bounds.
Let $\rho_m$ be the state in the ensemble with the largest probability $q_m$, and
$\{ \ket{e^m_j} | j=1,\cdots,d\}$ be the orthonormal eigenstates of $\rho_m$.
The classical cloning strategy could be performed with respect to the basis $\{ \ket{e^m_j} | j=1,\cdots,d\}$, therefore
$J(\mathcal{E})\geq\sum_{i}q_{i}F(\rho_{i},\rho_{i}')$,
where $\rho_{i}'=\sum_{j}\bra{e^m_j} \rho_{i}\ket{e^m_j}\proj{e^m_j}$
and $F(\rho_{m},\rho_{m}')=1$.
The fidelity satisfies the inequality $F(\rho,\rho')\geq\mathrm{tr}\rho\rho'$~\cite{jozsa}
for any two states $\rho$ and $\rho '$,
so $J(\mathcal{E})\geq\sum_{i\neq m}q_{i}\mathrm{tr}\rho_{i}\rho_{i}'+q_{m}=\sum_{i\neq m}q_{i}\mathrm{tr}\rho_{i}'^{2}+q_{m}$.
Since $\rho_i'$ is diagonal in the basis $\{ \ket{e^m_j} | j=1,\cdots,d\}$,
we have $\mathrm{tr}\rho_{i}'^{2}=\sum_{j=1}^{d}(\rho_i')_{jj}^{2}=\sum_{j=1}^{d}(\rho_i)_{jj}^{2}\geq(\sum_{j=1}^{d}(\rho_i)_{jj})^{2}/d=1/d$.
Thus $J(\mathcal{E})\geq\sum_{i\neq m}q_{i}/d+q_{m}=1/d+q_{m}(d-1)/d\geq (N+d-1)/Nd$
since $q_m \geq 1/N$.  This completes the proof.

\end{proof}

\section*{Appendix B: Proof of lemma 1}
\begin{proof}
A state $\rho$ in a two-dimensional Hilbert space has a spectral decomposition as $\rho=q_{0}\proj{\psi_{0}}+q_{1}\proj{\psi_{1}}$,
where $\ket{\psi_{0}}$ and $\ket{\psi_{1}}$ are the orthonormal
eigenstates of $\rho$.
We choose the basis for a classical cloning strategy as $\ket{e_{0}}=\cos(\theta_{1}/2)\ket{\psi_{0}}+\sin(\theta_{1}/2)e^{i\varphi_{1}}\ket{\psi_{1}}$
and $\ket{e_{1}}=\sin(\theta_{1}/2)\ket{\psi_{0}}-\cos(\theta_{1}/2)e^{i\varphi_{1}}\ket{\psi_{1}}$.
The output state from the classical cloning process is
\begin{equation}
\rho'=\bra{e_{0}}\rho\ket{e_{0}}\proj{e_{0}}+\bra{e_{1}}\rho\ket{e_{1}}\proj{e_{1}}.\end{equation}
As $\rho$ can be written as $\rho=(I+\boldsymbol{r}\cdot\boldsymbol{\sigma})/2$~\cite{nielsen},
where $\boldsymbol{r}$ is a real three-dimensional vector, $0\leq|\boldsymbol{r}|\leq1$,
and $\boldsymbol{\sigma}=(\sigma_{x},\sigma_{y},\sigma_{z})$. Similarly, we write $\rho'=(I+\boldsymbol{r}'\cdot\boldsymbol{\sigma})/2$.
Using eq. (10) given in ~\cite{jozsa}, we get \begin{equation}
\begin{split}F(\rho,\rho') & =\frac{1}{2}\{1+\boldsymbol{r}\cdot\boldsymbol{r}'+[(1-\boldsymbol{r}\cdot\boldsymbol{r})(1-\boldsymbol{r}'\cdot\boldsymbol{r}')]^{1/2}\}\\
 & =\frac{1}{2}\{1+(q_{0}-q_{1})^{2}\cos^{2}\theta_{1}\\
 & +[(1-(q_{0}-q_{1})^{2})(1-(q_{0}-q_{1})^{2}\cos^{2}\theta_{1})]^{1/2}\}.\end{split}
\end{equation}
 Let $\rho_{0}'=\sum_{i=0}^{1}|\langle e_{i}|\psi_{0}\rangle|^{2}\proj{e_{j}}$
and $\rho_{1}'=\sum_{j=0}^{1}|\langle e_{j}|\psi_{1}\rangle|^{2}\proj{e_{j}}$,
then \begin{equation}
\begin{split}\sum_{i=0}^{1}q_{i}F(\proj{\psi_{i}},\rho_{i}') & =\sum_{i,j=0}^{1}q_{i}|\langle e_{j}|\psi_{i}\rangle|^{4}\\
 & =1-\frac{1}{2}\sin^{2}\theta_{1}\end{split}
\end{equation}
Thus, \begin{equation}
\begin{split}F(\rho,\rho')- & \sum_{i=0}^{1}q_{i}F(\proj{\psi_{i}},\rho_{i}')\\
 & \geq\frac{1}{2}\sin^{2}\theta_{1}(1-(q_{0}-q_{1})^{2})\geq0\end{split}
.\end{equation}
This completes the proof of Lemma 1.  \end{proof}

\section*{Appendix C: Proof of theorem 2}

\begin{proof}
(i) For an ensemble $\mathcal{E}=\{q_{i},\ket{\psi_{i}}\}$ of pure states,
$J=max_{\{\ket j\}}\{\sum_{i}q_{i}F(\ket{\psi_{i}},\rho_{i}')\}=max_{\{\ket j\}}\{\sum_{ij}q_{i}|\langle j|\psi_{i}\rangle|^{4}\}$,
where $\rho_{i}'=\sum_{j}|\langle j|\psi_{i}\rangle|^{2}\proj j$.
 $\rho_{i}'$ is also the density matrix after the projective measurement
on $\ket{\psi_{i}}$ along the basis $\{\ket j\}$. By the definition
of $J$ we get $J\geq\overline{F_{ave}(\mathcal{E})}$, where the
average is over all projective measurements, with respect
to the unitarily invariant measure~\cite{qianren3}.
For any fixed state $\ket{\psi}$, one can prove that ~\cite{qianren2,jones}
\begin{equation}
\int|\langle\phi|\psi\rangle|^{2n}\mathrm{d}\Omega_{\phi}=\frac{\Gamma(d)\Gamma(1+n)}{\Gamma(1)\Gamma(d+n)}
\label{uim}
\end{equation}
where the integral is over all pure states $\phi$ in a $d$ dimensional  Hilbert space with respect to the unitarily invariant measure $\mathrm{d}\Omega_{\phi}$ on the pure states,
and $\Gamma(x)$ is the Gamma function.
So we get that \begin{equation}
\begin{split} & J\geq\overline{F_{ave}(\mathcal{E})}=\overline{\sum_{ij}q_{i}|\langle j|\psi_{i}\rangle|^{4}}\\
 & =d\sum_{i}q_{i}\int|\langle\phi|\psi_{i}\rangle|^{4}\mathrm{d}\Omega_{\phi}=d\sum_{i}q_{i}\frac{\Gamma(d)\Gamma(3)}{\Gamma(1)\Gamma(d+2)}\\
 & =\frac{2}{d+1}.\end{split}
\label{eq:4}\end{equation}
From above derivation, one can easily see that the lower bound is actually achieved by an infinite ensemble consisting equiprobably of all pure states in a $d$-dimensional Hilbert space.

(ii) For a two dimensional states ensemble, from Lemma 1, \begin{equation}
\begin{split}J= & \max_{\{\ket{e_{j}}\}}\{\sum_{ik}q_{i}F(\rho_{i},\rho_{i}')\}\\
 & \geq\max_{\{\ket{e_{j}}\}}\{\sum_{ik}q_{i}q_{ik}F(\ket{\psi_{ik}},\rho_{ik}')\}\\
 & =\max_{\{\ket{e_{j}}\}}\{\sum_{ik}q_{i}q_{ik}|\langle e_{j}|\psi_{ik}\rangle|^{4}\} \\
 &\geq  \sum_{ik}q_{i}q_{ik} \; \overline{ |\langle e_{j}|\psi_{ik}\rangle|^{4}},\end{split}
\end{equation}
where $\{q_{ik},\ket{\psi_{ik}}\}$ are the eigenvalues and corresponding
eigenvectors of $\rho_{i}$, and the average is over all projective measurements $\{\proj{e_j}\}$. From Eq. \eqref{uim}, we get
\begin{equation}
J \geq \frac{2}{3}\sum_{ik}q_{i}q_{ik}=\frac{2}{3}.\end{equation}
The lower bound is achieved by an infinite ensemble consisting equiprobably of all pure states on the Bloch sphere.
\end{proof}

\bibliographystyle{model1a-num-names}
\bibliographystyle{model1a-num-names}
\bibliography{<your-bib-database>}

\end{document}